\newif\ifarxiv
\arxivtrue 

\ifarxiv
	\documentclass[12pt]{article}
	\usepackage[margin=1in]{geometry}
	\usepackage[slantedGreek]{mathpazo}
	\usepackage[pdftex]{graphicx}
	\usepackage{amsfonts}
	\usepackage{amssymb}
	\usepackage{amsthm}
	\usepackage{graphicx,float}
	\usepackage{amsmath}%
	\usepackage{hyperref}
	\usepackage{natbib}
	\usepackage{booktabs}
\else
\fi
\usepackage{color,soul,marginnote}

\newtheorem{proposition}{Proposition}
\newtheorem{remark}{Remark}
\newtheorem{example}{Example}

\ifarxiv
\title{Synthetic Priors}
\author{
	\makebox[.4\linewidth]{Nick Polson}\\\textit{  Booth School of Business}\\\textit{  University of Chicago}\\\and
	\makebox[.4\linewidth]{Vadim Sokolov}\\\textit{ Department of Systems Engineering }\\	\textit{  and Operations Research}\\\textit{ George Mason University}
}
\date{First Draft: \\This Draft: \today}
\else
\fi

\graphicspath{{./fig/}}

\begin{document}
\ifarxiv
\maketitle
\begin{abstract}
Bayesian inference in generalized linear models requires a prior on the
coefficient vector $\beta$. Practitioners naturally reason about response
probabilities at specific covariate values, not about abstract log-odds
parameters. We develop \textit{synthetic priors}: informative Bayesian priors
for GLMs grounded in Good's \citeyearpar{good1950} device of imaginary
observations---the principle that every conjugate prior is equivalent to a
likelihood on pseudo-data from the same exponential family. The conditional
means prior of \citet{bedrick1996} elicits independent Beta priors on the
conditional mean response at $p$ expert-chosen design points; the induced prior
on $\beta$ is a product of binomial likelihoods at synthetic data points.
Combined with P\'{o}lya-Gamma data augmentation \citep{polson2013}, the posterior
admits an exact conjugate Gibbs sampler---no tuning, no Metropolis step---by
treating the augmented dataset as a standard logistic regression. We show that
ridge regression and catalytic priors \citep{huang2020} are instances of
Good's device, and identify prediction-powered inference
\citep{angelopoulos2023ppi} as a structural analogue in the frequentist
setting---all three mediate a variance--bias tradeoff through a single
informativeness parameter.
We illustrate the approach on two benchmark problems: the Challenger O-ring
data \citep{dalal1989}, where the BCJ prior provides a more moderate posterior
predictive at the 31°F launch temperature; and a Phase~II atopic dermatitis
dose-finding trial ($n = 300$), where the synthetic prior narrows 95\% credible
intervals by 3--6\% and raises decision probabilities by up to 2 percentage points
relative to a flat prior.
\end{abstract}
\else
\fi

\section{Introduction}
\label{sec:intro}

Prior elicitation is notoriously difficult.
Priors are the hallmark of Bayesian inference and can provide large efficiency
gains, but specifying them in applied work requires genuine domain expertise
that is rarely available on the natural parameter scale.
This difficulty is especially acute for logistic regression: a unit change in
$\beta_j$ corresponds to a multiplicative change in the odds, a scale that is
cognitively demanding to calibrate and sensitive to reparametrization.
Practitioners reason far more naturally about the probability of an event at a
specific covariate value: the placebo response rate in a clinical trial,
the default probability for a typical borrower, the click-through rate for a
reference advertisement.

\citet{bedrick1996} proposed the \textit{conditional means prior} (BCJ~prior):
choose $p$ design points $\tilde{x}_1, \ldots, \tilde{x}_p$ on the covariate
space, place independent Beta$(a_j, b_j)$ priors on the conditional mean
responses $\tilde{\mu}_j = \sigma(\tilde{x}_j^\top \beta)$, and derive the
implied prior on $\beta$ by the change of variables
$\beta = \tilde{X}^{-1}\operatorname{logit}(\tilde{\mu})$.
The key insight is that the induced prior on $\beta$ is exactly the likelihood of
$p$ synthetic binomial observations---$a_j$ successes out of $n_j = a_j + b_j$
trials at each design point $\tilde{x}_j$. Bayesian updating therefore reduces
to standard GLM computation on the combined real-plus-synthetic dataset.

The BCJ prior is \textit{predictive} in nature: elicitation happens on the
conditional mean scale $E[Y \mid X=x]$, the quantity a practitioner can assess
directly from prior trials or domain knowledge.
From an applied standpoint, one can often make such prior predictive
assessments naturally---for example, in clinical trials the LD50
(or ED50) is a standard target, and a practitioner can readily specify a
Beta prior on the probability of response at the maximum tolerated dose.
\citet{bedrick1996} illustrated this with the Challenger O-ring data;
we re-examine that example alongside a Phase~II atopic dermatitis dose-finding
trial to demonstrate the approach across settings.

We call this class \textit{synthetic priors}: priors for GLMs expressed as
likelihoods on pseudo-observations generated to encode expert knowledge.
This formulation has two immediate consequences. First, posterior computation is
exact: the augmented likelihood is a standard GLM likelihood, not an intractable
product. Second, P\'{o}lya-Gamma data augmentation \citep{polson2013} applies
without modification to the combined dataset, yielding an exact conjugate Gibbs
sampler with no tuning parameters.

The idea of encoding a prior as imaginary observations extends well beyond
the Gaussian linear model.  For ridge regression this connection is
well-known: the $\ell_2$ penalty is equivalent to $p$ imaginary observations
at the origin, inducing $N(0,\sigma^2/\lambda)$ on each coefficient.
We extend this perspective to the full exponential family---specifically
to logistic regression and GLMs more broadly---using Good's device of
imaginary observations \citep{good1950}.  Synthetic priors thereby unify
ridge regression with the BCJ construction under a single principle.

Synthetic priors also connect to two recent frameworks.  \citet{huang2020}
propose \textit{catalytic priors} that regularize by augmenting with data
generated from a simpler model's predictive distribution, tunable by a scalar
$\tau$. \citet{angelopoulos2023ppi} develop \textit{prediction-powered inference}
(PPI), a frequentist framework that combines a small labeled dataset with a large
ML-predicted dataset via a rectified estimator. All three share a common
structure: a primary dataset is augmented with an auxiliary source, and a single
scalar parameter controls the weight placed on that source.

This paper makes four contributions:
\begin{enumerate}
  \item We identify Good's device of imaginary observations as the theoretical
  foundation for synthetic priors: every conjugate prior is a likelihood on
  pseudo-data from the same exponential family.  Ridge regression and the BCJ
  prior are both instances of Good's device, the former at unit basis vectors
  with a Gaussian prior, the latter at expert-chosen design points with Beta
  priors.

  \item We show that combining the BCJ prior with P\'{o}lya-Gamma augmentation
  yields an exact conjugate Gibbs sampler in which the prior synthetic
  observations enter identically to real observations, with weight determined
  by the prior sample sizes $n_j = a_j + b_j$.

  \item We arrange ridge regression, the BCJ prior, catalytic priors, and
  prediction-powered inference into a hierarchy of synthetic prior frameworks,
  all mediating a variance--bias tradeoff through a single informativeness
  parameter.

  \item We demonstrate the approach in two applications: the Challenger
  O-ring data of \citet{bedrick1996}, where the BCJ prior moderates the
  flat posterior's extreme extrapolation at 31°F; and a Phase~II atopic
  dermatitis dose-finding trial, where a two-point BCJ prior narrows
  credible intervals by 3--6\% and raises decision probabilities by up to
  2 percentage points relative to a flat prior.
\end{enumerate}

Section~\ref{sec:bcj} develops synthetic priors via Good's device and the
BCJ conditional means prior, with the P\'{o}lya-Gamma Gibbs sampler in
Section~\ref{sec:pg}. Section~\ref{sec:unify} connects synthetic
priors to ridge regression, catalytic priors, and prediction-powered inference.
Section~\ref{sec:application} applies the approach to the Challenger
O-ring data and a Phase~II atopic dermatitis dose-finding trial. Section~\ref{sec:discussion} discusses extensions.

\section{Synthetic Priors via Conditional Predictive Means}
\label{sec:bcj}

\subsection{Good's Device and the BCJ Prior}
\label{sec:good}

\citet{good1950} proposed a principle for prior elicitation: every
prior should be interpretable as the posterior from a hypothetical past
experiment.  \citet{jeffreys1961} applied the same idea in specifying
priors for location--scale models, encoding prior knowledge as the result
of imaginary observations.
For exponential family models, this interpretation is exact.

\begin{proposition}[Prior--data equivalence]
\label{prop:good}
Let $L(\theta \mid y) = h(y)\exp\{\langle\theta, T(y)\rangle - A(\theta)\}$
be an exponential family likelihood with natural parameter $\theta$ and
sufficient statistic $T(y)$.  The conjugate prior
\[
  p(\theta \mid \nu_0, t_0) \propto \exp\bigl\{\langle\theta,\,\nu_0
  t_0\rangle - \nu_0 A(\theta)\bigr\}
\]
is proportional to the likelihood of $\nu_0$ observations with sufficient
statistic $t_0$.  The posterior after observing $n$ data points with total
sufficient statistic $T = \sum_{i=1}^n T(y_i)$ is the same conjugate family
with updated parameters $(\nu_0 + n,\; (\nu_0 t_0 + T)/(\nu_0 + n))$.
\end{proposition}

Good's device reframes hyperparameter choice as a concrete question:
\textit{how many imaginary observations, with what summary statistic,
would justify this degree of prior belief?}
The effective prior sample size $\nu_0$ quantifies confidence relative to
the anticipated data size $n$.

\begin{example}[Beta--Binomial]
Under $Y \mid \theta \sim \operatorname{Binomial}(n, \theta)$ with
$\theta \sim \operatorname{Beta}(\alpha, \beta)$, the posterior mean is
the precision-weighted average
\[
  \mathbb{E}[\theta \mid y] = \frac{\kappa}{\kappa + n}\,\mu
    + \frac{n}{\kappa + n}\,\hat\theta_{\mathrm{MLE}},
\]
where $\mu = \alpha/\kappa$ is the prior mean, $\kappa = \alpha + \beta$,
and $\hat\theta_{\mathrm{MLE}} = s/n$.  Proposition~\ref{prop:good} identifies
$\nu_0 = \kappa$ as the effective prior sample size: the prior is
equivalent to $\kappa$ imaginary trials with $\kappa\mu$ imaginary
successes.
\end{example}

The BCJ prior applies Good's device to logistic
regression: each design point $\tilde{x}_j$ induces an imaginary binomial
experiment with $n_j = a_j + b_j$ trials and success rate $a_j/n_j$.

Let $y_i \mid x_i, \beta \sim \operatorname{Bernoulli}(\mu_i)$ with
$\mu_i = \sigma(x_i^\top \beta)$, where $\sigma(t) = (1 + e^{-t})^{-1}$ is the
logistic function, $x_i \in \mathbb{R}^p$, and $\beta \in \mathbb{R}^p$.
The log-likelihood is
\begin{equation}
  \ell(\beta) = \sum_{i=1}^n \bigl[ y_i x_i^\top \beta
    - \log(1 + e^{x_i^\top \beta}) \bigr].
  \label{eq:loglik}
\end{equation}

Choose $p$ design points $\tilde{x}_1, \ldots, \tilde{x}_p \in \mathbb{R}^p$
such that the matrix $\tilde{X} = (\tilde{x}_1, \ldots, \tilde{x}_p)^\top$
is invertible. Place independent priors
\begin{equation}
  \tilde{\mu}_j \sim \operatorname{Beta}(a_j, b_j),
  \quad j = 1, \ldots, p,
  \label{eq:beta_prior}
\end{equation}
on the conditional mean responses $\tilde{\mu}_j = \sigma(\tilde{x}_j^\top\beta)$.
The joint prior on $\beta$ induced by the change of variables is:

\begin{equation}
  \pi(\beta) \propto \prod_{j=1}^p
    \sigma(\tilde{x}_j^\top\beta)^{a_j}
    \bigl[1 - \sigma(\tilde{x}_j^\top\beta)\bigr]^{b_j}.
  \label{eq:bcj}
\end{equation}

\begin{proposition}[Data augmentation representation, \citealt{bedrick1996}]
The prior \eqref{eq:bcj} is the likelihood of $p$ independent binomial
observations, each with $a_j$ successes out of $n_j = a_j + b_j$ trials at
design point $\tilde{x}_j$.
\end{proposition}

This is the synthetic-data interpretation: the expert's beliefs act as if a
previous experiment had been run with $n_j$ subjects at each design point, with
$a_j$ successes. The posterior combines evidence from this hypothetical
experiment with the observed data $\mathcal{D} = \{(y_i, x_i)\}_{i=1}^n$:
\begin{equation}
  \log \pi(\beta \mid \mathcal{D}) =
    \sum_{i=1}^n \bigl[ y_i x_i^\top \beta - \log(1 + e^{x_i^\top\beta})\bigr]
    + \sum_{j=1}^p \bigl[ a_j \tilde{x}_j^\top\beta
      - n_j \log(1 + e^{\tilde{x}_j^\top\beta})\bigr].
  \label{eq:logpost}
\end{equation}

The Beta parameters $(a_j, b_j)$ elicit directly on the probability scale:
$a_j / n_j$ is the prior mean response at $\tilde{x}_j$, and $n_j = a_j + b_j$
is the effective prior sample size at that design point. A practitioner who
believes the placebo response rate is approximately 10\%, with confidence
equivalent to 10 prior observations, sets $a_j = 1$ and $b_j = 9$.

\citet{hanson2014} develop an alternative observable-scale approach using
$g$-priors, in which a single scalar $g$ controls global shrinkage toward
an expert-specified baseline probability; the BCJ prior allows separate weights
$n_j$ at each design point and is therefore more flexible when beliefs differ
substantially across the covariate space.

The prior \eqref{eq:bcj} is proper whenever $\tilde{X}$ has full rank $p$ and
each $n_j \geq 1$ \citep{bedrick1996}. Properness does not require $p$ to equal
the number of covariates: one can augment with more than $p$ synthetic
observations, in which case the prior is over-determined and resembles a ridge
penalty centered at the expert mean.

\begin{remark}[Nonlinear parameters]
In a dose-response model $\text{logit}\,P(y|d) = \beta_0 + \beta_1 d$,
the ED50 (effective dose for 50\% response, known as the LD50 in toxicology)
satisfies $\text{ED50} = -\beta_0 / \beta_1$, a nonlinear function of $\beta$
and a key regulatory target in Phase~II dose-finding trials.
Placing a synthetic observation at the dose $d_{50}$ where the expected
response is 0.5 regularizes toward parameter configurations consistent
with the expert's ED50 belief, without specifying a prior directly on the
nonlinear quantity.
\end{remark}

\begin{remark}[Ridge regression as Good's device]
Ridge regression is the Gaussian linear analogue of the BCJ prior: the
$\ell_2$ penalty $\lambda\|\beta\|^2$ encodes $p$ imaginary observations at
the unit design points $e_j$ with response~0, inducing the prior
$\beta \sim N(0,\,\sigma^2/\lambda\cdot I_p)$.
Section~\ref{sec:ridge} develops this connection formally.
\end{remark}

\subsection{P\'{o}lya-Gamma Data Augmentation}
\label{sec:pg}

The augmented log-posterior \eqref{eq:logpost} is a logistic regression
log-likelihood on $n + p$ observations. P\'{o}lya-Gamma augmentation
\citep{polson2013} converts this into a conditionally Gaussian model.

\begin{proposition}[P\'{o}lya-Gamma Identity, \citealt{polson2013}]
\label{prop:pg}
For $b > 0$, $a \in \mathbb{R}$, and $\psi \in \mathbb{R}$,
\begin{equation}
  \frac{(e^\psi)^a}{(1 + e^\psi)^b}
  = 2^{-b}\, e^{\kappa \psi}
    \int_0^\infty e^{-\omega \psi^2 / 2}\, p(\omega;\, b,\, 0)\,d\omega,
  \label{eq:pg_identity}
\end{equation}
where $\kappa = a - b/2$ and $\omega \sim \operatorname{PG}(b, 0)$.
\end{proposition}

Conditioning on $\omega$, the logistic factor is a Gaussian function of
$\psi = x^\top\beta$:
\[
  \log p(y \mid x, \beta, \omega)
  = \kappa\, x^\top\beta - \tfrac{\omega}{2}(x^\top\beta)^2 + \text{const.}
\]
Combined with the Gaussian conjugate structure, this yields a two-step Gibbs
sampler.

\medskip
\noindent\textbf{Algorithm 1} (BCJ + P\'{o}lya-Gamma Gibbs).
\textit{Form the augmented dataset by stacking real and synthetic observations:
$X_{\mathrm{all}} = [X^\top, \tilde{X}^\top]^\top$,
$b_{\mathrm{all},i} = 1$ for real observations and
$b_{\mathrm{all},j} = n_j$ for synthetic observations, and
$\kappa_{\mathrm{all},i} = y_i - \tfrac{1}{2}$ for real observations and
$\kappa_{\mathrm{all},j} = a_j - n_j/2$ for synthetic observations.
Initialize $\beta^{(0)}$. At iteration $t$:}
\begin{enumerate}
  \item[\textit{Step 1.}] \textit{Sample latent variables.} For each row $i$
  of $X_{\mathrm{all}}$:
  \[
    \omega_i^{(t)} \mid \beta^{(t-1)} \sim
    \operatorname{PG}\!\bigl(b_{\mathrm{all},i},\;
    x_{\mathrm{all},i}^\top \beta^{(t-1)}\bigr).
  \]

  \item[\textit{Step 2.}] \textit{Sample coefficients.} Let
  $\Omega^{(t)} = \operatorname{diag}(\omega_1^{(t)}, \ldots)$. Then
  \[
    \beta^{(t)} \mid \Omega^{(t)} \sim
    \mathcal{N}\!\bigl(m_n,\; V_n\bigr),
  \]
  where
  \[
    V_n^{-1} = X_{\mathrm{all}}^\top \Omega^{(t)} X_{\mathrm{all}},
    \qquad
    m_n = V_n\, X_{\mathrm{all}}^\top\, \kappa_{\mathrm{all}}.
  \]
\end{enumerate}

Each step is an exact draw; no Metropolis correction is required.
The precision matrix $V_n^{-1}$ decomposes additively as
\[
  V_n^{-1} = \underbrace{X^\top \Omega_{\mathrm{real}} X}_{\text{real data}}
           + \underbrace{\tilde{X}^\top \tilde{\Omega} \tilde{X}}_{\text{BCJ prior}},
\]
revealing that the BCJ prior contributes a positive-semidefinite increment to
the information matrix whose magnitude is determined by the prior sample sizes
$n_j$.

In R, exact P\'{o}lya-Gamma samples are drawn using \texttt{BayesLogit::rpg()}
\citep{polson2013}. The sampler requires no adaptation period and converges
rapidly in practice because the P\'{o}lya-Gamma variables are conditionally
independent given $\beta$.

\section{Connections with Previous Work}
\label{sec:unify}

Synthetic priors connect to several existing frameworks.
Ridge regression, the BCJ prior, catalytic priors \citep{huang2020}, and
prediction-powered inference \citep{angelopoulos2023ppi} all augment a primary
inference problem with an auxiliary source of information, with a scalar
parameter controlling the weight on that source.
Table~\ref{tab:compare} summarizes the hierarchy.

\subsection{Ridge Regression: The Gaussian Linear Case}
\label{sec:ridge}

The simplest synthetic prior arises in linear regression.
Suppose $y_i = x_i^\top\beta + \varepsilon_i$ with $\varepsilon_i \sim N(0,\sigma^2)$.
Ridge regression \citep{hoerl1970} adds $\ell_2$ penalization:
\begin{equation}
  \hat\beta_{\mathrm{ridge}} = \operatorname*{arg\,min}_\beta
    \bigl\{\|y - X\beta\|^2 + \lambda\|\beta\|^2\bigr\}.
  \label{eq:ridge}
\end{equation}

\begin{proposition}[Ridge as augmented OLS]
\label{prop:ridge}
$\hat\beta_{\mathrm{ridge}}$ equals the ordinary least-squares estimator on
the augmented dataset
$\bar{X} = [X^\top,\; \sqrt{\lambda}\,I_p]^\top$,
$\bar{y} = [y^\top,\; 0_p^\top]^\top$.
\end{proposition}

\begin{proof}
The augmented objective is
$\|\bar{y} - \bar{X}\beta\|^2 = \|y - X\beta\|^2 + \lambda\|\beta\|^2$,
which matches \eqref{eq:ridge} term by term.
\end{proof}

The $p$ augmented rows $(\sqrt{\lambda}\,e_j,\; 0)$ are synthetic
observations encoding the prior belief $\beta \sim N(0,\,(\sigma^2/\lambda)I_p)$.
By Proposition~\ref{prop:good}, each augmented row is one imaginary
observation at the unit design point $e_j$ with response 0, contributing
precision $\lambda$ per coordinate.  The precision matrix decomposes as
\[
  \bar{X}^\top\bar{X} = X^\top X + \lambda I_p
  = \underbrace{X^\top X}_{\text{real data}}
  + \underbrace{\lambda I_p}_{\text{prior}},
\]
identically to the BCJ decomposition in the P\'{o}lya-Gamma sampler
(Section~\ref{sec:pg}), with $\lambda I_p$ playing the role of
$\tilde{X}^\top\tilde\Omega\tilde{X}$.

\begin{table}[ht]
  \centering
  \caption{The synthetic prior hierarchy.}
  \label{tab:compare}
  \begin{tabular}{@{}lp{2cm}p{2cm}p{2cm}p{2.2cm}@{}}
    \toprule
    & Ridge (1970) & BCJ (1996) & Catalytic (2020) & PPI (2023) \\
    \midrule
    Framework       & Bayesian MAP & Bayesian    & Bayesian      & Frequentist \\
    Likelihood      & Gaussian     & Logistic    & GLM           & Any \\
    Synthetic data  & $N(0,\tau^2)$ at $e_j$ & Expert at $\tilde{x}_j$ & Simpler model & ML predictions \\
    Informativeness & $\lambda$    & $n_j$       & $\tau^{-1}$   & $\hat{\lambda}$ \\
    Nonlinear params & No          & Yes         & No            & No \\
    Validity        & MAP consist. & Post.\ consist. & Post.\ consist. & Coverage at finite $n$ \\
    \bottomrule
  \end{tabular}
\end{table}

\subsection{Catalytic Priors and Prediction-Powered Inference}

\citet{huang2020} construct the catalytic prior by fitting a simpler model $g$
(e.g., an intercept-only GLM) to the data, generating synthetic covariates
$X^* = (x_1^*, \ldots, x_M^*)$ from the observed covariate distribution, and
drawing synthetic responses $Y_j^*$ from $g$'s predictive distribution. The
prior is
\begin{equation}
  \pi_{\mathrm{cat}}(\beta) \propto
  \prod_{j=1}^M f(Y_j^* \mid x_j^*, \beta)^{1/\tau}.
  \label{eq:catalytic}
\end{equation}
The hyperparameter $\tau > 0$ downweights the synthetic observations; large
$\tau$ yields a diffuse, weakly informative prior. In linear regression,
straightforward algebra shows that \eqref{eq:catalytic} reduces to ridge
regression with the penalty centered at $\hat{\beta}_0$ from the simpler model
rather than at zero.

The structural difference from BCJ is the source of pseudo-data: catalytic priors
derive $Y^*$ from a fitted model, making the prior data-adaptive but not
genuinely informative in the sense of encoding external knowledge. BCJ priors
derive from expert elicitation and are informative by design.

\medskip
\noindent\textbf{Prediction-Powered Inference.}
\citet{angelopoulos2023ppi} address a frequentist problem: given a small labeled
dataset $\mathcal{D}_L$ and a large unlabeled dataset $\mathcal{D}_U$ with
ML-predicted labels $\hat{Y}$, construct valid confidence intervals for a
population parameter $\theta$. The PPI estimator is
\begin{equation}
  \hat{\theta}^{\mathrm{PPI}} =
  \hat{\theta}_U^{\mathrm{pred}}
  + \underbrace{\bigl(\hat{\theta}_L^{\mathrm{true}}
    - \hat{\theta}_L^{\mathrm{pred}}\bigr)}_{\text{rectifier}},
  \label{eq:ppi}
\end{equation}
where the rectifier debiases the large-sample estimate using a labeled subset.
PPI guarantees frequentist coverage regardless of ML model quality; when the ML
model is perfect, the rectifier vanishes and the estimator uses only
$\mathcal{D}_U$.

The PPI++ extension \citep{angelopoulos2023ppipp} selects the optimal scalar
weight $\hat{\lambda}$ on the rectifier by minimizing asymptotic variance,
making the variance--bias tradeoff explicit.

\subsection{The Control Variate Perspective}

All three frameworks admit a \textit{control-variate} interpretation: the
auxiliary source (expert beliefs, simpler model, ML predictions) provides a
high-variance but low-cost estimate, and the scalar weight ($n_j$, $1/\tau$,
$\hat{\lambda}$) controls how much variance reduction to accept at the cost of
possible bias. The frameworks differ in three dimensions: (i)~whether validity
requires the auxiliary source to be accurate (PPI: no; BCJ and catalytic: yes,
asymptotically); (ii)~whether the framework is Bayesian (BCJ, catalytic) or
frequentist (PPI); (iii)~whether the auxiliary source can supply information
from outside the observed covariate distribution (BCJ: yes; catalytic and PPI:
no).

\section{Application: Atopic Dermatitis Phase II Dose-Finding}
\label{sec:application}

\subsection{Clinical Setting}

We apply the BCJ + P\'{o}lya-Gamma sampler to a simulated Phase~II
dose-finding trial for a hypothetical IL-inhibitor for moderate-to-severe
atopic dermatitis. The design follows the DoseFinding R package's canonical
binary endpoint vignette \citep{bretz2005}.

\begin{itemize}
  \item \textbf{Endpoint.} IGA (Investigator Global Assessment) binary response
  (IGA 0/1 vs.\ 2--4).
  \item \textbf{Doses.} 0 (placebo), 0.5, 1.5, 2.5, 4.0~mg subcutaneous q2w.
  \item \textbf{Sample size.} $n = 60$ per arm, $N = 300$ total.
  \item \textbf{True model.} Emax on the logit scale:
  $\operatorname{logit} P(y=1 \mid d) =
  \operatorname{logit}(0.10) + [\operatorname{logit}(0.35) - \operatorname{logit}(0.10)]
  \cdot d/(0.5 + d)$.
\end{itemize}

The true placebo response rate is 10\% and the true top-dose response rate is
approximately 31\%.

\subsection{BCJ Prior Specification}

Prior beliefs are calibrated to published Phase~3 trial data on IL-4/13
inhibitors \citep{simpson2016}. The placebo response rate of 10\% matches
dupilumab trial placebo arms directly; the top-dose response of 30\% is set
below the dupilumab 300~mg q2w rate ($\approx$38\%) to reflect a hypothetical
agent with lower efficacy. Two design points anchor the dose-response:

\begin{table}[ht]
  \centering
  \caption{BCJ prior design points for the atopic dermatitis trial.}
  \label{tab:prior}
  \begin{tabular}{@{}lllll@{}}
    \toprule
    Design point & Prior & Prior mean & Weight \\
    \midrule
    Dose $= 0$~mg (placebo) & Beta$(1, 9)$ & $10\%$ & 10 pseudo-obs \\
    Dose $= 4$~mg (top dose) & Beta$(3, 7)$ & $30\%$ & 10 pseudo-obs \\
    \bottomrule
  \end{tabular}
\end{table}

The two synthetic observations enter the Gibbs sampler as binomial
observations: $(n_j = 10, a_j = 1)$ at placebo and $(n_j = 10, a_j = 3)$ at
the top dose. The P\'{o}lya-Gamma parameters are $b_j = n_j$ for synthetic
observations and $b_i = 1$ for real observations.

\subsection{Results}

The Gibbs sampler was run for 10,000 iterations with 3,000 burn-in; P\'{o}lya-Gamma
variates were drawn using the \texttt{BayesLogit} package \citep{polson2013}. Mean effective sample sizes
over 30 replicates were 3,590 for $\beta_0$ and 4,355 for $\beta_1$,
indicating low autocorrelation; Geweke $z$-scores were within $[-2,2]$ in all
30 replicates for both parameters, indicating no evidence of non-convergence.

Table~\ref{tab:results} reports posterior predicted IGA response rates averaged
over 30 independent replications.
Both priors yield nearly identical posterior means: BCJ and flat differ by at
most 0.3~pp in mean response at any dose, since 10~pseudo-observations per
design point represent approximately 3\% of the total effective information
(10 of 310 observations at each anchor).
The BCJ advantage is in variance reduction: 95\% credible interval widths are
3--6\% narrower, with the largest reduction at the top dose (6\%).
Both priors overestimate response at placebo (15.5--15.8\% vs.\ true 10\%) and
underestimate it at mid-doses; this is model misspecification bias from fitting
a linear logistic model to an Emax data-generating process.

\begin{table}[ht]
  \centering
  \caption{Posterior predicted IGA response rates, averaged over 30 simulated
  trials (SD in parentheses). CrI width is the mean width of the 95\% posterior
  credible interval.}
  \label{tab:results}
  \begin{tabular}{@{}llllll@{}}
    \toprule
    Dose (mg) & True $P$ & BCJ Post.\ Mean & Flat Post.\ Mean & BCJ CrI Width & Flat CrI Width \\
    \midrule
    0.0 & $10.0\%$ & $15.5\%\;(2.8)$ & $15.8\%\;(3.0)$ & $11.5\,\text{pp}$ & $12.0\,\text{pp}$ \\
    0.5 & $19.7\%$ & $17.1\%\;(2.6)$ & $17.4\%\;(2.8)$ & $10.6\,\text{pp}$ & $11.0\,\text{pp}$ \\
    1.5 & $26.6\%$ & $21.0\%\;(2.2)$ & $21.3\%\;(2.4)$ & $ 9.3\,\text{pp}$ & $ 9.6\,\text{pp}$ \\
    2.5 & $29.3\%$ & $25.6\%\;(2.2)$ & $25.9\%\;(2.4)$ & $10.6\,\text{pp}$ & $11.1\,\text{pp}$ \\
    4.0 & $31.1\%$ & $33.8\%\;(3.4)$ & $34.1\%\;(3.9)$ & $19.4\,\text{pp}$ & $20.6\,\text{pp}$ \\
    \bottomrule
  \end{tabular}
\end{table}

Table~\ref{tab:decision} reports decision probabilities averaged over 30
replicates.

\begin{table}[ht]
  \centering
  \caption{Decision probabilities
  $P\bigl(P(y=1 \mid d) - P(y=1 \mid 0) > 5\%\bigr)$,
  averaged over 30 simulated trials.}
  \label{tab:decision}
  \begin{tabular}{@{}lll@{}}
    \toprule
    Dose (mg) & BCJ & Flat \\
    \midrule
    0.5 & $0.000$ & $0.000$ \\
    1.5 & $0.621$ & $0.610$ \\
    2.5 & $0.894$ & $0.872$ \\
    4.0 & $0.954$ & $0.937$ \\
    \bottomrule
  \end{tabular}
\end{table}

The BCJ prior raises decision probabilities by 1.2--2.2 percentage points at
doses of 1.5~mg and above, with the largest absolute gain at 2.5~mg ($+0.022$).
The fitted linear logistic model underestimates mid-dose response rates relative
to the Emax data-generating process (BCJ mean 21.0\% at 1.5~mg versus true
26.6\%), since the linear link cannot capture the concavity of the Emax curve on
the logit scale; both priors are subject to this specification bias.
The results support a minimum effective dose between 1.5 and 2.5~mg.

Figure~\ref{fig:results} summarises both results. The left panel shows that
BCJ and flat posterior means are nearly identical and follow a plausible
dose-response shape; both underestimate mid-dose response relative to the true
Emax curve, reflecting the linear logistic model misspecification; BCJ credible
bands are visibly narrower at all doses.
The right panel shows that BCJ decision probabilities exceed the flat-prior
values at doses of 1.5~mg and above, with the minimum effective dose threshold
(80\% probability) crossed between 1.5 and 2.0~mg under both priors.

\begin{figure}[ht]
  \centering
  \includegraphics[width=0.48\textwidth]{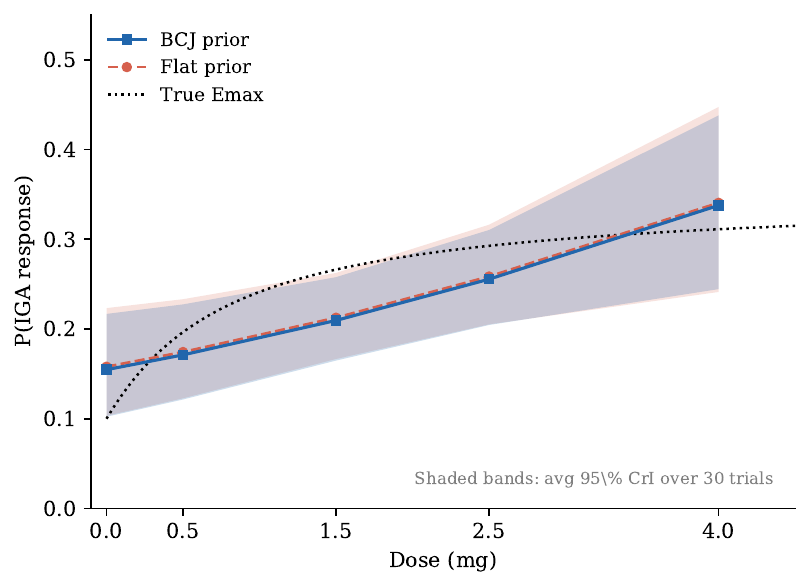}%
  \hfill
  \includegraphics[width=0.48\textwidth]{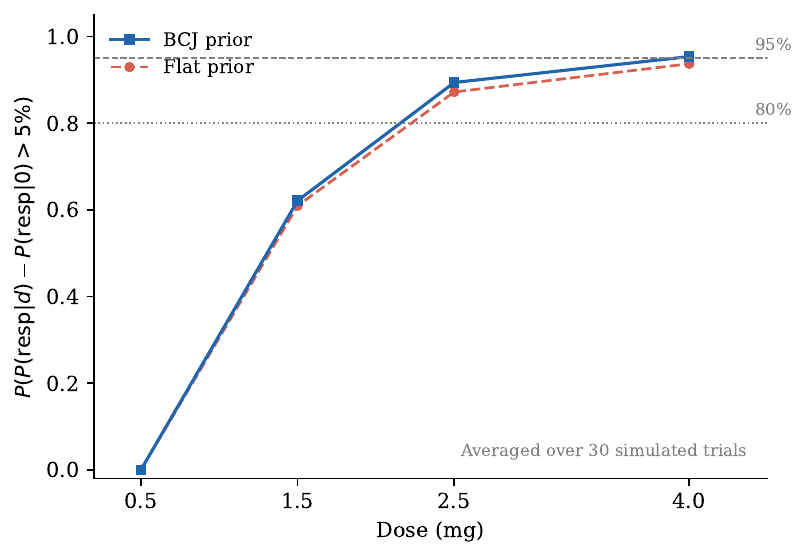}
  \caption{Atopic dermatitis dose-finding simulation, averaged over 30 trials.
  \textit{Left}: Posterior mean dose-response curve with average 95\% credible
  bands for BCJ (blue) and flat (orange dashed) priors; dotted black curve is
  the true Emax model.
  \textit{Right}: Posterior probability of clinically meaningful benefit over
  placebo ($>5$ percentage points) as a function of dose;
  horizontal reference lines at 80\% and 95\%.}
  \label{fig:results}
\end{figure}

\subsection{O-Ring Failure: Reproducing the BCJ Benchmark}
\label{sec:oring}

\citet{bedrick1996} illustrated the conditional means prior with the
Challenger Space Shuttle O-ring data \citep{dalal1989}: 23~pre-Challenger
launches recorded the temperature at liftoff (°F) and whether at least one
primary field-joint O-ring failed.  We re-examine this example using the
P\'{o}lya-Gamma Gibbs sampler.

\medskip
\noindent\textbf{Model.}
$\operatorname{logit} P(\text{failure} \mid \text{temp}) = \beta_0 + \beta_1
\cdot \text{temp}$, with $n = 23$ binary observations ($y=1$ for 7 launches
with at least one failure).

\medskip
\noindent\textbf{BCJ prior.}
Two design points anchor the temperature--failure curve:
\begin{itemize}
  \item Temp $= 31$°F (Challenger launch day, very cold):
    Beta$(8,2)$ $\Rightarrow$ prior mean $80\%$, weight $n_j=10$
  \item Temp $= 81$°F (warm, near maximum observed):
    Beta$(1,9)$ $\Rightarrow$ prior mean $10\%$, weight $n_j=10$
\end{itemize}
These beliefs are consistent with engineering assessments that O-ring
resilience degrades sharply at low temperatures.

\medskip
\noindent\textbf{Results.}
Table~\ref{tab:oring} reports posterior predicted failure probabilities at
selected temperatures from 10{,}000 Gibbs iterations (3{,}000 burn-in),
using the same P\'{o}lya-Gamma Gibbs sampler with a single fixed seed
(code in \texttt{code/oring\_analysis.R}).

\begin{table}[ht]
  \centering
  \caption{Posterior predicted O-ring failure probability at selected
  temperatures, Challenger data.}
  \label{tab:oring}
  \begin{tabular}{@{}lllll@{}}
    \toprule
    Temp (°F) & BCJ Post.\ Mean & BCJ 95\% CrI & Flat Post.\ Mean & Flat 95\% CrI \\
    \midrule
    31 (Challenger) & $87.3\%$ & $[62\%, 99\%]$ & $98.8\%$ & $[89\%, 100\%]$ \\
    50              & $64.9\%$ & $[41\%, 87\%]$ & $94.4\%$ & $[66\%, 100\%]$ \\
    65              & $36.4\%$ & $[20\%, 55\%]$ & $57.0\%$ & $[28\%, 85\%]$ \\
    75              & $20.7\%$ & $[ 9\%, 37\%]$ & $11.4\%$ & $[ 1\%, 32\%]$ \\
    81              & $14.3\%$ & $[ 4\%, 30\%]$ & $ 4.4\%$ & $[ 0\%, 21\%]$ \\
    \bottomrule
  \end{tabular}
\end{table}

At the Challenger launch temperature of 31°F, the BCJ posterior mean is
$87.3\%$ (95\% CrI: 62--99\%), versus $98.8\%$ (95\% CrI: 89--100\%) under
a flat prior.  The BCJ prior moderates the flat posterior's extreme estimate
by anchoring at the engineer's prior mean of 80\% at 31°F, while remaining
consistent with a high probability of failure.
At intermediate temperatures (50--65°F), the BCJ prior substantially narrows
the flat prior's wide credible intervals: the flat upper bound drops from
100\% to 87\% at 50°F and from 85\% to 55\% at 65°F, as both BCJ anchors
jointly constrain the feasible slope of the temperature--failure curve.
At warm temperatures (75--81°F), the BCJ prior pulls the posterior mean up
to 14--21\% versus the flat posterior's 4--11\%, reflecting the steeper
slope required to accommodate the cold-temperature anchor.
The BCJ credible interval at 81°F is also wider (4--30\%) than the flat
prior's (0--21\%), since the 31°F anchor forces a slope that maintains
non-trivial failure probability even at high temperatures.
Both posteriors agree that the Challenger launch carried a very high
O-ring failure risk.

\section{Discussion}
\label{sec:discussion}

Synthetic priors embed expert knowledge as pseudo-observations in the GLM
likelihood. The BCJ formulation elicits on the probability scale, which is
natural for practitioners. The P\'{o}lya-Gamma Gibbs sampler handles the
augmented dataset without modification, and the synthetic observations enter
identically to real observations with weight $n_j = a_j + b_j$.

Several extensions are immediate. First, the framework generalizes to
multinomial and ordered logit models: the conditional means prior places Dirichlet
priors on category probabilities at chosen covariate profiles, and generalized
P\'{o}lya-Gamma augmentation \citep{windle2014} provides the Gibbs sampler.
Second, the informativeness parameter $n_j$ can be shared across design points
to implement a globally calibrated prior strength, analogous to the scalar
$\tau$ in catalytic priors. Third, in sequential trial settings, earlier interim
data can update the synthetic prior between stages---the updated BCJ distribution
becomes the prior for the next stage, with posterior sample sizes replacing the
prior sample sizes $n_j$.

The connection to prediction-powered inference highlights a limitation of
synthetic priors: unlike PPI, they do not provide frequentist validity guarantees
when the expert beliefs are wrong. If the true placebo rate were 25\% rather
than 10\%, the BCJ prior with Beta$(1,9)$ at placebo would be misspecified and
could bias all posterior inferences. PPI's rectifier corrects for ML model error
automatically; BCJ priors have no analogous self-correcting mechanism. In
practice, sensitivity to the prior specification can be assessed by varying $n_j$
from 0 (flat prior) to larger values and inspecting whether conclusions are
stable---a standard Bayesian robustness check.

\bibliography{ref}
\end{document}